\documentclass{article}

\usepackage{arxiv}

\usepackage[utf8]{inputenc}    
\usepackage[T1]{fontenc}       

\usepackage{amsmath,amsfonts,amssymb}
\usepackage{amsthm}
\usepackage{array}
\usepackage{booktabs}
\usepackage{graphicx}
\usepackage{url}
\usepackage{enumitem}
\usepackage{algpseudocode}
\usepackage{tikz}
\usetikzlibrary{arrows.meta, positioning, calc, shapes.geometric, fit, backgrounds}
\usepackage{xcolor}
\usepackage{cite}
\usepackage{lettrine}          
\usepackage{hyperref}          
\newcommand{\IEEEPARstart}[2]{\lettrine[lines=2]{#1}{#2}}

\theoremstyle{plain}
\newtheorem{theorem}{Theorem}
\newtheorem{proposition}{Proposition}
\newtheorem{lemma}{Lemma}
\newtheorem{corollary}{Corollary}
\theoremstyle{definition}
\newtheorem{definition}{Definition}

\theoremstyle{remark}
\newtheorem{remark}{Remark}

\newcounter{alg}
\newcommand{\algheader}[1]{%
  \refstepcounter{alg}%
  \vskip1pt\hrule\vskip3pt%
  \noindent\textbf{Algorithm~\thealg:}\ #1\par\vskip3pt\hrule\vskip3pt}

\newcommand{\NN}{\mathbb{N}}
\newcommand{\Bit}{\{0,1\}}
\newcommand{\Bits}{\{0,1\}^{*}}
\newcommand{\G}{\mathsf{G}}
\newcommand{\gray}{\mathsf{gray}}
\newcommand{\V}{V}
\newcommand{\Lev}{d_{\mathrm{L}}}
\newcommand{\Ham}{d_{\mathrm{H}}}
\newcommand{\Len}{L}
\newcommand{\eps}{\varepsilon}
\newcommand{\xor}{\oplus}
\newcommand{\Ab}{\{0,1,\dots,b-1\}}
\newcommand{\Abs}{\{0,1,\dots,b-1\}^{*}}
\newcommand{\shift}[1]{\widehat{#1}}

\tikzset{
  wordnode/.style={draw=black, thick, rounded corners=3pt, fill=blue!8,
    minimum width=1.15cm, minimum height=0.62cm, font=\ttfamily\small, inner sep=2pt},
  emptynode/.style={wordnode, fill=gray!12},
  subEdge/.style={-{Stealth[scale=1.0]}, blue!55!black, thick},
  insEdge/.style={-{Stealth[scale=1.1]}, red!70!black, thick, dashed},
  blabel/.style={font=\footnotesize\sffamily, text=gray!60!black}
}

\title{Variable-length Gray codes for the Natural Numbers}

\author{
	\href{https://orcid.org/0000-0001-8231-5687}{\includegraphics[scale=0.06]{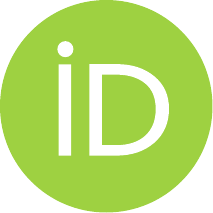}\hspace{1mm}Ezequiel L\'opez-Rubio} \\
	Department of Computer Languages and Computer Science \\
	University of M\'alaga \\
	Bulevar Louis Pasteur 35, 29071 M\'alaga, Spain \\
	and ITIS Software, University of M\'alaga \\
	C/ Arquitecto Francisco Pe\~nalosa 18, 29010 M\'alaga, Spain \\
	\texttt{ezeqlr@lcc.uma.es} \\
}

\date{}

\renewcommand{\shorttitle}{Variable-length Gray codes for the Natural Numbers}

\hypersetup{
pdftitle={Variable-length Gray codes for the Natural Numbers},
pdfauthor={Ezequiel L\'opez-Rubio},
pdfkeywords={Gray code, variable-length codes, bijective numeration, Levenshtein distance, b-ary codes, universal integer codes, edit distance, large language models, tokenization, numeracy},
}

\begin{document}
\maketitle
\begin{abstract}
The modular $b$-ary Gray code arranges  fixed-length $b$-ary representations of  intervals of natural numbers so that consecutive numbers differ in a single digit. Its usefulness, however, is tied to a fixed word length. We introduce, for every integer base $b\ge 2$, a \emph{variable-length Gray code} $\V_{b}$: a bijection from the natural numbers onto the set $\Abs$ of all finite strings over the $b$-symbol alphabet $\{0,1,\dots,b-1\}$. The construction orders the natural numbers by codeword length into blocks---block $k$ being the $b^{k}$ strings of length $k$ and beginning at index $N_{k}=(b^{k}-1)/(b-1)=1+b+\cdots+b^{k-1}$---and lists each block along the modular $b$-ary Gray code of the within-block offset, with its leading digit incremented modulo $b$. We prove four properties for arbitrary $b$. First, $\V_{b}$ is a bijection onto $\Abs$, a complete code that assigns exactly one codeword to every finite string over the alphabet, including the empty string. Second, the Levenshtein (edit) distance between the codewords of two consecutive integers is always one: a single-symbol substitution inside a length block, and a single leading-symbol insertion $0^{k}\mapsto1\,0^{k}$ at each block boundary. Third, codeword length is monotone non-decreasing in the encoded number and equals $\lfloor\log_{b}((b-1)n+1)\rfloor$. Fourth, the transition to codewords of length $k$ occurs exactly at $n=N_{k}=\sum_{i=0}^{k-1}b^{i}$, where $\V_{b}(n)=1\,0^{k-1}$. The binary code is the case $b=2$, in which the modular binary Gray code coincides with the reflected binary Gray code $m\mapsto m\oplus\lfloor m/2\rfloor$. The presented code therefore realizes a Hamiltonian enumeration of $\Abs$ under unit edit steps while retaining a near-optimal, self-adapting length profile. These properties suggest a use in machine learning: large language models routinely emit natural numbers as symbol strings, yet the positional notations they rely on are neither complete---strings with leading zeros are invalid or redundant---nor locally stable, since incrementing a number may rewrite many symbols at once. Because $\V_{b}$ is a complete code over $\Abs$ and moves by a single edit between consecutive integers, it removes both obstacles and is a natural candidate representation for numeric tokens; we develop this argument and discuss the code's compression behavior relative to fixed-length Gray codes.
\end{abstract}


\section{Introduction}
\label{sec:intro}

\IEEEPARstart{T}{he} binary Gray code is one of the most widely used non-standard numeral systems in computer science and engineering \cite{gray1953pulse,gilbert1958gray,doran2007gray}. Its defining property is that the binary words assigned to two numerically consecutive integers differ in exactly one bit position. This single-bit-change property makes Gray codes indispensable wherever a physical or logical quantity is sampled while it changes: rotary and linear position encoders, analog-to-digital converters, asynchronous state machines, Karnaugh maps, and the mutation neighborhoods of genetic algorithms all exploit it to avoid the large, transient errors that ordinary binary counting produces when several bits flip at once.

The classical Gray code is, however, a \emph{fixed-length} object over the \emph{binary} alphabet. There are two distinct extensions of the binary Gray code to arbitrary bases $b$, namely the \emph{reflected} and \emph{modular} $b$-ary Gray codes. To encode natural numbers in a base $b$ one must commit in advance to a word length $w$; the code then represents only the $b^{w}$ natural numbers $0,1,\dots,b^{w}-1$ (for the binary case, $b=2$), and it spends the full $w$ symbols on every one of them, including small numbers that would need far fewer. Two limitations follow. First, the representable range is bounded, so an unbounded stream of natural numbers cannot be encoded without periodically enlarging $w$ and, in general, re-encoding. Second, the representation is wasteful for skewed distributions in which small natural numbers dominate, a situation common when integers denote indices, counters, offsets, or repetition factors inside a larger structured object.

A concrete and timely instance of the second limitation arises in the training of large language models (LLMs). Such models generate text densely populated with natural numbers---counts, indices, identifiers, keys, offsets, and repetition factors---which they must emit one symbol at a time \cite{thawani2021representing,spithourakis2018numeracy}. The surface form in which a number is written has been shown to have a strong effect on how well a model learns to manipulate it \cite{nogueira2021arithmetic,wallace2019numbers,singh2024tokenization}. Standard positional notation is ill-suited to this setting in two respects that a Gray-type code addresses directly. First, it is not \emph{complete}: over a fixed alphabet the valid canonical representations are only a sparse subset of all strings---those without leading zeros---so a generative model must expend capacity learning to avoid producing invalid or redundant forms rather than modelling the quantities themselves. Second, it is not \emph{locally stable}: incrementing a number can rewrite many symbols at once (in base ten, $999\mapsto1000$ changes four digits and the length), so numerically adjacent values may have distant representations. A code that is complete over all strings and changes by a single edit between consecutive integers removes both obstacles at once; the variable-length Gray code introduced here has exactly these two properties, and Section~\ref{sec:discussion} argues in detail why they make it an attractive representation for numeric tokens. This application motivates, but is logically independent of, the mathematical development that follows.

In this paper we remove both restrictions at once, for every integer base $b\ge 2$. We define a \emph{variable-length Gray code}
\[
  \V_{b} : \NN \longrightarrow \Abs ,
\]
which maps each natural number to a finite string over the alphabet $\{0,1,\dots,b-1\}$ of self-adapting length. The construction is elementary and orders the natural numbers by codeword length. The empty string encodes $0$; the $b$ strings of length one encode $1,\dots,b$; the $b^{2}$ strings of length two come next; and so on, block $k$---the length-$k$ strings---beginning at $N_{k}=1+b+\cdots+b^{k-1}=(b^{k}-1)/(b-1)$. Within a block we arrange the strings along the modular $b$-ary Gray code of the within-block offset and rotate the leading symbol by one modulo $b$, so that each block runs from $1\,0^{k-1}$ to $0^{k}$ and successive blocks are joined by inserting a single leading symbol, $0^{k}\mapsto1\,0^{k}$. Tables~\ref{tab:example2}, \ref{tab:example10} and~\ref{tab:example3} list the first codewords for $b=2$, $b=10$, and $b=3$.

We establish four properties, stated informally here and proved in Section~\ref{sec:properties} for arbitrary $b\ge 2$. (i)~$\V_{b}$ is a bijection onto the set of \emph{all} finite strings over $\{0,\dots,b-1\}$, so every string is a codeword of exactly one integer and the code is complete. (ii)~The Levenshtein distance \cite{levenshtein1966binary} between the codewords of $n$ and $n+1$ is always one, so the enumeration $\V_{b}(0),\V_{b}(1),\V_{b}(2),\dots$ is a single-edit walk that visits every finite string exactly once. (iii)~Codeword length is monotone non-decreasing in $n$ and equals $\lfloor\log_{b}((b-1)n+1)\rfloor$, which is within one symbol of the information-theoretic ideal $\log_{b}$ of the magnitude. (iv)~The passage to codewords of length $k$ happens precisely at $n=N_{k}=\sum_{i=0}^{k-1}b^{i}$, at which point the codeword is a one followed by $k-1$ zeros.

For $b=2$ the modular binary Gray code and the reflected binary code coincide. They are known as the binary Gray code $m\mapsto m\oplus\lfloor m/2\rfloor$, and $\V_{2}$ is exactly the binary Gray-code analogue of the ordinary bijective binary numeration that sends $n$ to the binary representation of $n+1$ with its leading one removed; the general $b$ replaces the binary Gray code by its modular $b$-ary counterpart. Taken together, the four properties characterize $\V_{b}$ as a Gray code for the countably infinite set $\Abs$ ordered by unit edit distance, rather than for a fixed-size $b$-ary cube ordered by unit Hamming distance. 

The remainder of the paper is organized as follows. Section~\ref{sec:previous} reviews Gray codes, edit distances, and variable-length integer codes. Section~\ref{sec:method} fixes notation, recalls the standard fixed-length $b$-ary Gray code, and defines $\V_{b}$, with worked examples for $b=2$ and $b=10$. Section~\ref{sec:properties} proves the four theorems for arbitrary $b\ge 2$. Section~\ref{sec:discussion} discusses compression relative to fixed-length Gray codes, and Section~\ref{sec:conclusion} concludes.

\section{Previous Works}
\label{sec:previous}

\paragraph{Gray codes}
The reflected binary code was patented by Gray~\cite{gray1953pulse} for pulse-code communication, although the underlying single-change ordering had been used earlier in telegraphy and puzzles. Gilbert~\cite{gilbert1958gray} placed it on a combinatorial footing by identifying Gray codes with Hamiltonian paths on the $n$-dimensional hypercube, where vertices are $n$-bit strings and edges join strings at Hamming distance one. The standard $b$-bit reflected Gray code $\G(\cdot,b)$ is the bijection between $\{0,\dots,2^{b}-1\}$ and $\Bit^{b}$ given by $\gray(m)=m\xor\lfloor m/2\rfloor$; it is a single Hamiltonian path of the hypercube, and consecutive codewords therefore have unit Hamming distance. Beyond the binary code, a large literature studies \emph{combinatorial} Gray codes, in which the objects of a combinatorial class (combinations, permutations, partitions, spanning trees, and many others) are listed so that successive objects differ by a minimal local change; the survey by Savage~\cite{savage1997survey} and the treatment in Knuth's \emph{Art of Computer Programming} \cite{knuth2011taocp4a} give a comprehensive account. Our construction belongs to this tradition, but the underlying combinatorial class is the set of \emph{all} finite binary strings and the local change is a single \emph{edit} operation rather than a bit flip; the fixed length of the classical code is replaced by a variable one.

\paragraph{Edit and Hamming distances}
The Hamming distance \cite{hamming1950error} counts positions at which two equal-length strings differ and is the natural metric of fixed-length Gray codes. When strings may have different lengths, the appropriate metric is the Levenshtein (edit) distance \cite{levenshtein1966binary}, the minimum number of single-character insertions, deletions, and substitutions transforming one string into another. Hamming distance is the restriction of Levenshtein distance to equal-length strings and to substitution operations only; a unit Hamming distance is therefore also a unit Levenshtein distance. This relationship is what allows a variable-length code to inherit, in a suitable sense, the single-change property of the classical Gray code.

\paragraph{Variable-length and universal integer codes}
Encoding the natural numbers by binary strings of number-dependent length is the subject of a rich body of work on universal codes. Elias~\cite{elias1975universal} introduced the $\gamma$, $\delta$, and $\omega$ codes, self-delimiting representations whose length grows like $\log_2 n$ up to lower-order terms. Golomb~\cite{golomb1966run} gave run-length codes optimal for geometric distributions, later specialized as Rice codes. Apostolico and Fraenkel~\cite{apostolico1987robust} and Fraenkel and Klein~\cite{fraenkel1996robust} developed complete and robust universal codes based on Fibonacci (Zeckendorf) representations, prized for their resynchronization properties. A recurring design axis in this literature is whether the code is \emph{prefix-free} (self-delimiting) or \emph{bijective}. Bijective numeration systems, discussed for instance by Salomaa~\cite{salomaa1981jewels} and Knuth~\cite{knuth2011taocp4a}, place the finite strings over an alphabet in one-to-one correspondence with the natural numbers; the familiar binary instance maps $n$ to the binary representation of $n+1$ with its leading one removed. Our code is the Gray-coded analogue of exactly this bijection, generalized to an arbitrary base $b$: within each length block we replace ordinary $b$-ary counting by the $b$-ary Gray code. The novelty is that this single change endows the bijection with a unit-edit-distance property between consecutive integers, which ordinary bijective base-$b$ numeration does not possess. The trade-off between length optimality and the various structural guarantees of such codes is treated from an information-theoretic standpoint by Cover and Thomas~\cite{coverthomas2006}.

\paragraph{Number representation in language models}
How numbers are written on the page matters for the neural sequence models that read and generate them. Wallace et al.~\cite{wallace2019numbers} showed that standard token and sub-word embeddings capture numerical magnitude only imperfectly, and Spithourakis and Riedel~\cite{spithourakis2018numeracy} found that ordinary language-model output heads predict numerals poorly and proposed dedicated numeric mechanisms. Nogueira et al.~\cite{nogueira2021arithmetic} demonstrated that the \emph{surface form} of a number strongly determines whether a transformer learns to add and subtract, with sub-word segmentation performing worst and explicit position markers best; Singh and Strouse~\cite{singh2024tokenization} traced arithmetic errors in frontier models to their number-tokenization schemes, noting that production systems disagree---some tokenizing digit by digit, others in one-, two-, or three-digit chunks. The survey by Thawani et al.~\cite{thawani2021representing} organizes these representational choices and their trade-offs. A common thread is that the default sub-word tokenizers inherited from general text \cite{sennrich2016neural} treat numerals as if they were words, an inductive bias poorly matched to numeric structure. The present construction speaks to this literature from an unusual angle: rather than augmenting an embedding with magnitude information, it changes the \emph{string} assigned to each integer so that the assignment is complete (no invalid or redundant forms) and locally smooth (unit edit distance between successors). We develop the implications in Section~\ref{sec:discussion}.

\section{Methodology}
\label{sec:method}

\subsection{Notation}
Fix an integer base $b\ge 2$. Let $\NN=\{0,1,2,\dots\}$ be the natural numbers including zero, and let $\Ab$ be the $b$-symbol alphabet, whose elements we call \emph{digits}. Let $\Abs$ be the set of all finite strings (words) over this alphabet, including the empty string $\eps$. For a word $w$, $\Len(w)$ denotes its length and $\Len(\eps)=0$; we write $0^{k}$ for the all-zero word of length $k$ (so that $0^{0}=\eps$) and use juxtaposition for concatenation, so that $1\,0^{k}$ is the digit $1$ followed by $k$ zeros. Digit arithmetic is taken modulo $b$ where indicated. The binary alphabet $\Bit$ and the binary strings $\Bits$ are the case $b=2$.

For two words $u,v\in\Abs$, the \emph{Levenshtein distance} $\Lev(u,v)$ is the minimum number of single-character insertions, deletions, and substitutions needed to transform $u$ into $v$ \cite{levenshtein1966binary}. When $\Len(u)=\Len(v)$, the \emph{Hamming distance} $\Ham(u,v)$ is the number of positions at which they differ \cite{hamming1950error}. For equal-length words, $\Lev(u,v)\le\Ham(u,v)$, since every mismatch may be repaired by one substitution; in particular $\Ham(u,v)=1$ implies $\Lev(u,v)=1$.

\subsection{The Standard Fixed-Length $b$-ary Gray Code}
Fix a word length $w\ge 1$. For $0\le m<b^{w}$, let $a_{w-1}a_{w-2}\cdots a_{0}$ be the $w$-digit base-$b$ representation of $m$, with $a_{w-1}$ the most significant digit. The \emph{$w$-digit $b$-ary Gray code} of $m$ is the word $\G(m,w)=\gamma_{w-1}\gamma_{w-2}\cdots\gamma_{0}\in\Ab^{w}$ defined by
\begin{equation}
  \gamma_{w-1}=a_{w-1}, \qquad
  \gamma_{i}=(a_{i}-a_{i+1})\bmod b\quad (0\le i\le w-2),
  \label{eq:graydef}
\end{equation}
that is, each Gray digit is the difference, modulo $b$, between a base-$b$ digit and the next more significant one. The map is inverted by $a_{w-1}=\gamma_{w-1}$ and $a_{i}=(\gamma_{i}+a_{i+1})\bmod b$. For $b=2$, \eqref{eq:graydef} reads $\gamma_{i}=\beta_{i+1}\xor\beta_{i}$, the reflected binary Gray code, equivalently $\gray(m)=m\xor\lfloor m/2\rfloor$. Algorithm~\ref{alg:gray} evaluates \eqref{eq:graydef} in a single pass over the digits, from the most significant downward, and is used as a subroutine by the encoder of Section~\ref{sec:varlen}. We record the standard facts we shall use \cite{gilbert1958gray,doran2007gray,knuth2011taocp4a}.

\begin{figure}[!t]
\algheader{The $b$-ary Gray word $\G(m,w)=\gamma_{w-1}\cdots\gamma_{0}$ from the base-$b$ digits $a_{w-1}\cdots a_{0}$ of $m$, evaluating~\eqref{eq:graydef}.}\label{alg:gray}
\begin{algorithmic}[1]
\Procedure{GrayWord}{$b,\ a_{w-1}a_{w-2}\cdots a_{0}$}
  \State $\gamma_{w-1}\gets a_{w-1}$ \Comment{copy the most significant digit}
  \For{$i\gets w-2$ \textbf{downto} $0$}
    \State $\gamma_{i}\gets (a_{i}-a_{i+1})\bmod b$ \Comment{difference from next more significant digit}
  \EndFor
  \State \Return $\gamma_{w-1}\gamma_{w-2}\cdots\gamma_{0}$
\EndProcedure
\end{algorithmic}
\vskip1pt\hrule
\end{figure}

\begin{proposition}[Basic $b$-ary Gray-code facts]
\label{prop:gray}
Let $b\ge 2$ and $w\ge 1$.
\begin{enumerate}[label=\textup{(P\arabic*)},leftmargin=2.6em]
  \item \textup{(Bijection.)} The map $m\mapsto\G(m,w)$ is a bijection from $\{0,\dots,b^{w}-1\}$ onto $\Ab^{w}$.
  \item \textup{(Adjacency.)} For $0\le m<b^{w}-1$, the words $\G(m,w)$ and $\G(m+1,w)$ differ in exactly one position, and at that position by $+1$ modulo $b$; in particular $\Ham\!\big(\G(m,w),\G(m+1,w)\big)=1$.
  \item \textup{(Leading digit.)} The most significant digit of $\G(m,w)$ equals that of $m$.
  \item \textup{(Widening.)} If $w'>w$ and $0\le m<b^{w}$, then $\G(m,w')=0^{\,w'-w}\,\G(m,w)$.
  \item \textup{(Endpoints.)} $\G(0,w)=0^{w}$ and $\G(b^{w}-1,w)=(b-1)\,0^{\,w-1}$.
\end{enumerate}
\end{proposition}

\begin{proof}
(P1) The displayed inverse recovers the base-$b$ digits of $m$ from $\G(m,w)$, so the map is a bijection between the $b^{w}$ integers and the $b^{w}$ words of $\Ab^{w}$. (P2) is the defining property of the modular $b$-ary Gray code: adding one to $m$ changes exactly one Gray digit, and by $+1$ modulo $b$ \cite[\S7.2.1.1]{knuth2011taocp4a}. (P3) is immediate from~\eqref{eq:graydef}, since $\gamma_{w-1}=a_{w-1}$. For (P4), if $m<b^{w}$ then in width $w'$ the digits $a_{w},\dots,a_{w'-1}$ are zero, so $\gamma_{w'-1}=a_{w'-1}=0$ and $\gamma_{i}=(a_{i}-a_{i+1})\bmod b=0$ for $w\le i\le w'-2$, while $\gamma_{w-1}=(a_{w-1}-a_{w})\bmod b=a_{w-1}$; hence the extra high-order positions are zeros and the low-order block coincides with $\G(m,w)$. For (P5), $m=0$ gives $a_{i}=0$ for all $i$, so $\gamma_{i}=0$ and $\G(0,w)=0^{w}$; and $m=b^{w}-1$ gives $a_{i}=b-1$ for all $i$, so $\gamma_{w-1}=b-1$ and $\gamma_{i}=((b-1)-(b-1))\bmod b=0$ for $i\le w-2$, whence $\G(b^{w}-1,w)=(b-1)\,0^{\,w-1}$.
\end{proof}

\begin{remark}
\label{rem:modular}
Two $b$-ary Gray codes are standard: the \emph{reflected} code and the \emph{modular} code~\eqref{eq:graydef} used here \cite[\S7.2.1.1]{knuth2011taocp4a}. They coincide for $b=2$. We use the modular code because its endpoints are $0^{w}$ and $(b-1)\,0^{\,w-1}$ for \emph{every} $b$ \textup{(}property~\textup{(P5))}, whereas the reflected code ends at $(b-1)\,0^{\,w-1}$ only when $b$ is even and at $(b-1)^{w}$ when $b$ is odd. As Section~\ref{sec:varlen} shows, it is exactly~(P5) that makes every length block terminate at the all-zero word and lets consecutive blocks be joined by a single insertion for all bases at once.
\end{remark}

\subsection{The Variable-Length $b$-ary Gray Code}
\label{sec:varlen}
We order the natural numbers by codeword length. For $k\ge 0$ put
\[
  N_{k}=\sum_{i=0}^{k-1}b^{i}=\frac{b^{k}-1}{b-1}, \qquad N_{0}=0,
\]
so that $N_{0}<N_{1}<N_{2}<\cdots$ and $N_{k+1}-N_{k}=b^{k}$. The block
\[
  B_{k}=\{\,n\in\NN : N_{k}\le n<N_{k+1}\,\}
\]
therefore contains exactly $b^{k}$ integers, and the blocks $B_{0},B_{1},B_{2},\dots$ partition $\NN$. Every $n\in\NN$ has a unique representation
\[
  n=N_{k}+r, \qquad 0\le r<b^{k},
\]
in which $k$ is its \emph{block index} and $r$ its \emph{offset} within the block.

For a nonempty word $w=w_{k-1}w_{k-2}\cdots w_{0}\in\Ab^{k}$ let $\shift{w}$ denote the word obtained by incrementing its leading digit modulo $b$,
\[
  \shift{w}=\big((w_{k-1}+1)\bmod b\big)\,w_{k-2}\cdots w_{0},
\]
and set $\shift{\eps}=\eps$.

\begin{definition}[Variable-length $b$-ary Gray code]
\label{def:V}
For $n\in\NN$ with block index $k$ and offset $r=n-N_{k}$, the \emph{variable-length $b$-ary Gray code} of $n$ is
\[
  \V_{b}(n)=\shift{\G(r,k)}\in\Ab^{k},
\]
the $k$-digit $b$-ary Gray code of the offset, with its leading digit incremented modulo $b$. In particular $\V_{b}(0)=\shift{\G(0,0)}=\eps$.
\end{definition}

By endpoint property~(P5), $\G(0,k)=0^{k}$ and $\G(b^{k}-1,k)=(b-1)\,0^{\,k-1}$, so the leading increment sends the first word of block $k$ to $\V_{b}(N_{k})=\shift{0^{k}}=1\,0^{\,k-1}$ and the last word to $\V_{b}(N_{k+1}-1)=\shift{(b-1)0^{\,k-1}}=0^{k}$. Each block thus runs from $1\,0^{\,k-1}$ to $0^{k}$; and since the first word of the next block is $\V_{b}(N_{k+1})=1\,0^{k}$, consecutive blocks are joined by inserting a single leading digit, $0^{k}\mapsto1\,0^{k}$. Within a block, consecutive offsets yield Gray words that differ in a single position by~(P2), and the common leading increment preserves this; Section~\ref{sec:properties} makes these statements precise.

Algorithm~\ref{alg:vb} turns Definition~\ref{def:V} into a procedure that maps $n$ to $\V_{b}(n)$: it locates the block index $k$ and offset $r$ from the cumulative sizes $N_{k}=(b^{k}-1)/(b-1)$, forms the offset's Gray word $\G(r,k)$ by calling Algorithm~\ref{alg:gray}, and finally applies $\shift{\cdot}$ by incrementing the leading digit modulo $b$. Its cost is linear in the codeword length $k=\lfloor\log_{b}((b-1)n+1)\rfloor$ established in Section~\ref{sec:properties}.

\begin{figure}[!t]
\algheader{The variable-length $b$-ary Gray code $\V_{b}(n)=\shift{\G(r,k)}$ of Definition~\ref{def:V}.}\label{alg:vb}
\begin{algorithmic}[1]
\Require base $b\ge 2$ and a natural number $n\ge 0$
\Ensure the codeword $\V_{b}(n)\in\Abs$
\If{$n=0$}
  \State \Return $\eps$ \Comment{the empty word encodes $0$}
\EndIf
\State $k\gets 1$
\While{$(b^{k+1}-1)/(b-1)\le n$} \Comment{block index: $N_{k}\le n<N_{k+1}$}
  \State $k\gets k+1$
\EndWhile
\State $r\gets n-(b^{k}-1)/(b-1)$ \Comment{offset within block, $0\le r<b^{k}$}
\State $a_{k-1}a_{k-2}\cdots a_{0}\gets$ base-$b$ digits of $r$, padded to length $k$
\State $\gamma_{k-1}\gamma_{k-2}\cdots\gamma_{0}\gets\Call{GrayWord}{b,\ a_{k-1}a_{k-2}\cdots a_{0}}$ \Comment{$\G(r,k)$, Algorithm~\ref{alg:gray}}
\State $\gamma_{k-1}\gets(\gamma_{k-1}+1)\bmod b$ \Comment{$\shift{\cdot}$: increment the leading digit}
\State \Return $\gamma_{k-1}\gamma_{k-2}\cdots\gamma_{0}$
\end{algorithmic}
\vskip1pt\hrule
\end{figure}

\begin{remark}
For $b=2$ the code~\eqref{eq:graydef} is the reflected binary Gray code and $\shift{\cdot}$ flips the leading bit. In this case $\V_{2}(n)$ equals the reflected Gray code of $n+1$ with its leading zeros and then its leading one deleted---the Gray-coded bijective binary numeration mentioned in Section~\ref{sec:intro}. Tables~\ref{tab:example2}, \ref{tab:example10}, \ref{tab:example3} and Figs.~\ref{fig:walk2}, \ref{fig:walk10}, \ref{fig:walk3} illustrate $b=2$, $b=10$, and $b=3$ (the last two shown up to $n=1000$).
\end{remark}

\begin{table}[!t]
  \centering
  \caption{The first codewords of the binary code $\V_{2}$ (the case $b=2$). For $n$ in block $k$ the offset is $r=n-N_{k}$ with $N_{k}=2^{k}-1$; $\G(r,k)$ is the $k$-bit binary Gray code of $r$, and $\V_{2}(n)=\shift{\G(r,k)}$ increments its leading bit. Horizontal rules mark the block boundaries; each block runs from $1\,0^{\,k-1}$ to $0^{k}$.}
  \label{tab:example2}
  \footnotesize
  \begin{tabular}{@{}c c c@{\hspace{2em}}c c c@{}}
    \toprule
    $n$ & $\G(r,k)$ & $\V_{2}(n)$ & $n$ & $\G(r,k)$ & $\V_{2}(n)$ \\
    \midrule
    $0$ & $\eps$        & $\eps$       & $8$  & \texttt{001}  & \texttt{101} \\
    \cmidrule(r){1-3}
    $1$ & \texttt{0}    & \texttt{1}   & $9$  & \texttt{011}  & \texttt{111} \\
    $2$ & \texttt{1}    & \texttt{0}   & $10$ & \texttt{010}  & \texttt{110} \\
    \cmidrule(r){1-3}
    $3$ & \texttt{00}   & \texttt{10}  & $11$ & \texttt{110}  & \texttt{010} \\
    $4$ & \texttt{01}   & \texttt{11}  & $12$ & \texttt{111}  & \texttt{011} \\
    $5$ & \texttt{11}   & \texttt{01}  & $13$ & \texttt{101}  & \texttt{001} \\
    $6$ & \texttt{10}   & \texttt{00}  & $14$ & \texttt{100}  & \texttt{000} \\
    \cmidrule(lr){1-3}\cmidrule(lr){4-6}
    $7$ & \texttt{000}  & \texttt{100} & $15$ & \texttt{0000} & \texttt{1000} \\
    \bottomrule
  \end{tabular}
\end{table}

\begin{table}[!t]
  \centering
  \caption{The decimal code $\V_{10}$ ($b=10$) up to $n=1000$, with vertical ellipses compressing each block. Here $N_{k}=(10^{k}-1)/9$, so blocks $1,2,3$ are $\{1,\dots,10\}$, $\{11,\dots,110\}$, $\{111,\dots,1110\}$; $\G(r,k)$ is the $k$-digit decimal Gray code of the offset $r=n-N_{k}$, and $\V_{10}(n)=\shift{\G(r,k)}$ increments its leading digit modulo $10$. Each block runs from $1\,0^{\,k-1}$ to $0^{k}$; $n=1000$ lies inside block~$3$.}
  \label{tab:example10}
  \footnotesize
  \begin{tabular}{@{}r l l@{}}
    \toprule
    $n$ & $\G(r,k)$ & $\V_{10}(n)$ \\
    \midrule
    $0$    & $\eps$       & $\eps$ \\
    \midrule
    $1$    & \texttt{0}   & \texttt{1} \\
    $2$    & \texttt{1}   & \texttt{2} \\
           & $\vdots$     & $\vdots$ \\
    $10$   & \texttt{9}   & \texttt{0} \\
    \midrule
    $11$   & \texttt{00}  & \texttt{10} \\
           & $\vdots$     & $\vdots$ \\
    $110$  & \texttt{90}  & \texttt{00} \\
    \midrule
    $111$  & \texttt{000} & \texttt{100} \\
           & $\vdots$     & $\vdots$ \\
    $1000$ & \texttt{801} & \texttt{901} \\
           & $\vdots$     & $\vdots$ \\
    \bottomrule
  \end{tabular}
\end{table}

\begin{table}[!t]
  \centering
  \caption{The ternary code $\V_{3}$ ($b=3$) up to $n=1000$, with vertical ellipses compressing each block. Here $N_{k}=(3^{k}-1)/2$, so the block starts are $N_{1},\dots,N_{6}=1,4,13,40,121,364$; $\G(r,k)$ is the $k$-digit ternary Gray code of the offset $r=n-N_{k}$, and $\V_{3}(n)=\shift{\G(r,k)}$ increments the leading digit modulo $3$. Each block runs from $1\,0^{\,k-1}$ to $0^{k}$; $n=1000$ lies inside block~$6$.}
  \label{tab:example3}
  \footnotesize
  \begin{tabular}{@{}r l l@{\hspace{1.4em}}r l l@{}}
    \toprule
    $n$ & $\G(r,k)$ & $\V_{3}(n)$ & $n$ & $\G(r,k)$ & $\V_{3}(n)$ \\
    \midrule
    $0$  & $\eps$       & $\eps$       & $40$   & \texttt{0000}   & \texttt{1000} \\
    \cmidrule(r){1-3}
    $1$  & \texttt{0}   & \texttt{1}   &        & $\vdots$        & $\vdots$ \\
    $2$  & \texttt{1}   & \texttt{2}   & $120$  & \texttt{2000}   & \texttt{0000} \\
    \cmidrule(l){4-6}
    $3$  & \texttt{2}   & \texttt{0}   & $121$  & \texttt{00000}  & \texttt{10000} \\
    \cmidrule(r){1-3}
    $4$  & \texttt{00}  & \texttt{10}  &        & $\vdots$        & $\vdots$ \\
         & $\vdots$     & $\vdots$     & $363$  & \texttt{20000}  & \texttt{00000} \\
    \cmidrule(l){4-6}
    $12$ & \texttt{20}  & \texttt{00}  & $364$  & \texttt{000000} & \texttt{100000} \\
    \cmidrule(r){1-3}
    $13$ & \texttt{000} & \texttt{100} &        & $\vdots$        & $\vdots$ \\
         & $\vdots$     & $\vdots$     & $1000$ & \texttt{221211} & \texttt{021211} \\
    $39$ & \texttt{200} & \texttt{000} &        & $\vdots$        & $\vdots$ \\
    \bottomrule
  \end{tabular}
\end{table}

\begin{figure}[!t]
  \centering
  \resizebox{10cm}{!}{%
  \begin{tikzpicture}[node distance=6mm and 20mm]
    \node[blabel] (h0) at (0,1.05) {$\Len=0$};
    \node[blabel] (h1) at (2.5,1.05) {$\Len=1$};
    \node[blabel] (h2) at (5.0,1.05) {$\Len=2$};
    \node[blabel] (h3) at (7.5,1.05) {$\Len=3$};

    \node[emptynode] (e)  at (0,0.2) {$\eps$};
    \node[wordnode] (a1) at (2.5,0.2)  {1};
    \node[wordnode] (a0) at (2.5,-0.55) {0};
    \node[wordnode] (b10) at (5.0,0.2)   {10};
    \node[wordnode] (b11) at (5.0,-0.55) {11};
    \node[wordnode] (b01) at (5.0,-1.30) {01};
    \node[wordnode] (b00) at (5.0,-2.05) {00};
    \node[wordnode] (c100) at (7.5,0.2)   {100};
    \node[wordnode] (c101) at (7.5,-0.55) {101};
    \node[wordnode] (c111) at (7.5,-1.30) {111};
    \node[wordnode] (c110) at (7.5,-2.05) {110};
    \node[wordnode] (c010) at (7.5,-2.80) {010};
    \node[wordnode] (c011) at (7.5,-3.55) {011};
    \node[wordnode] (c001) at (7.5,-4.30) {001};
    \node[wordnode] (c000) at (7.5,-5.05) {000};

    \draw[subEdge] (a1) -- (a0);
    \draw[subEdge] (b10) -- (b11);
    \draw[subEdge] (b11) -- (b01);
    \draw[subEdge] (b01) -- (b00);
    \draw[subEdge] (c100) -- (c101);
    \draw[subEdge] (c101) -- (c111);
    \draw[subEdge] (c111) -- (c110);
    \draw[subEdge] (c110) -- (c010);
    \draw[subEdge] (c010) -- (c011);
    \draw[subEdge] (c011) -- (c001);
    \draw[subEdge] (c001) -- (c000);

    \draw[insEdge] (e)   to[out=-25,in=180] (a1);
    \draw[insEdge] (a0)  to[out=-25,in=180] (b10);
    \draw[insEdge] (b00) to[out=-25,in=180] (c100);

    \node[blabel, anchor=west] at (-0.8,-4.0) {\textcolor{blue!55!black}{$\rightarrow$} substitution};
    \node[blabel, anchor=west] at (-0.8,-4.6) {\textcolor{red!70!black}{$\dashrightarrow$} insertion of a leading \texttt{1}};
  \end{tikzpicture}%
  }%
  \caption{The enumeration $\V_{2}(0),\V_{2}(1),\V_{2}(2),\dots$ for $b=2$ as a walk on $\Bits$ under unit Levenshtein steps. Each column is a length block; within a column, consecutive codewords differ by one substitution (Theorem~\ref{thm:lev}, interior case), and each block boundary is a single insertion of a leading one, taking $0^{k}$ to $1\,0^{k}$ (Theorems~\ref{thm:lev} and~\ref{thm:step}). The walk visits every finite binary string exactly once (Theorem~\ref{thm:bij}).}
  \label{fig:walk2}
\end{figure}
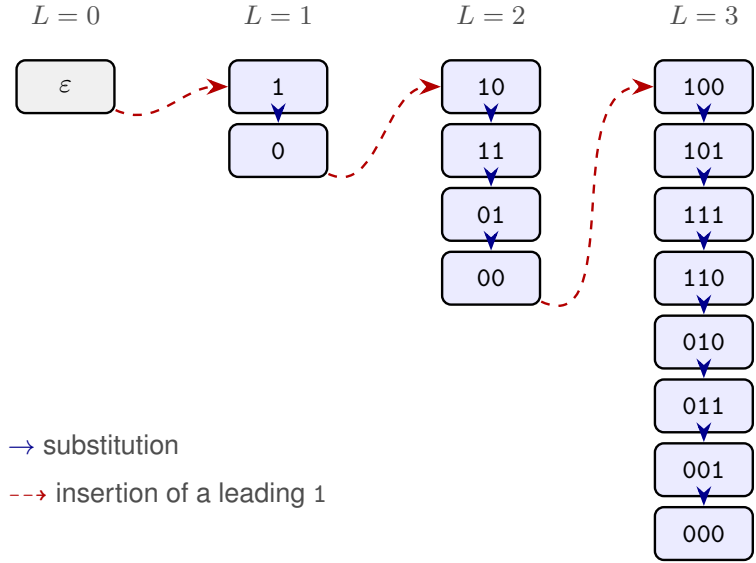

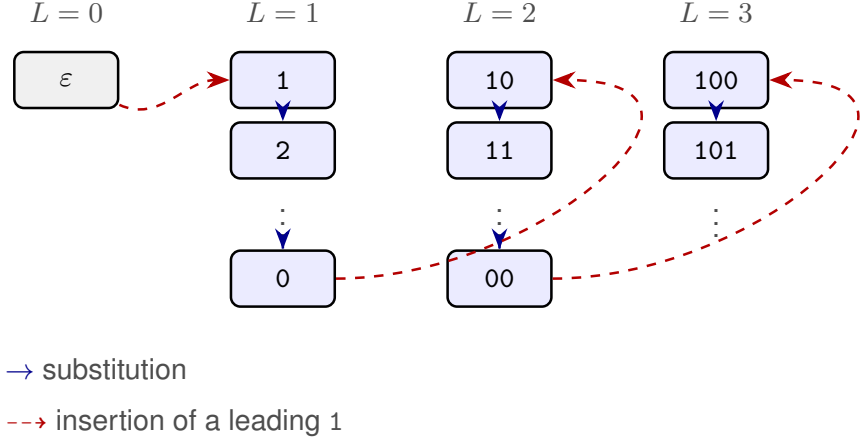
\begin{figure}[!t]
  \centering
  \resizebox{13cm}{!}{%
  \begin{tikzpicture}
    \node[blabel] at (0,0.95)   {$\Len=0$};
    \node[blabel] at (2.4,0.95) {$\Len=1$};
    \node[blabel] at (4.8,0.95) {$\Len=2$};
    \node[blabel] at (7.2,0.95) {$\Len=3$};
    \node[emptynode] (e) at (0,0.2) {$\eps$};
    \node[wordnode] (n1) at (2.4,0.20)  {1};
    \node[wordnode] (n2) at (2.4,-0.58) {2};
    \node[blabel]   (nd) at (2.4,-1.28) {$\vdots$};
    \node[wordnode] (n0) at (2.4,-2.00) {0};
    \node[wordnode] (m10) at (4.8,0.20)  {10};
    \node[wordnode] (m11) at (4.8,-0.58) {11};
    \node[blabel]   (md)  at (4.8,-1.28) {$\vdots$};
    \node[wordnode] (m00) at (4.8,-2.00) {00};
    \node[wordnode] (p100) at (7.2,0.20)  {100};
    \node[wordnode] (p101) at (7.2,-0.58) {101};
    \node[blabel]   (pd)   at (7.2,-1.28) {$\vdots$};
    \draw[subEdge] (n1)--(n2); 
    \draw[subEdge] (nd)--(n0);
    \draw[subEdge] (m10)--(m11); 
    \draw[subEdge] (md)--(m00);
    \draw[subEdge] (p100)--(p101); 
    
    \draw[insEdge] (e)   to[out=-25,in=180] (n1);
    \draw[insEdge] (n0)  to[out=0,in=0,looseness=1.8] (m10);
    \draw[insEdge] (m00) to[out=0,in=0,looseness=1.8] (p100);
    \node[blabel, anchor=west] at (-0.8,-3.0) {\textcolor{blue!55!black}{$\rightarrow$} substitution};
    \node[blabel, anchor=west] at (-0.8,-3.6) {\textcolor{red!70!black}{$\dashrightarrow$} insertion of a leading \texttt{1}};
  \end{tikzpicture}%
  }%
  \caption{The enumeration $\V_{10}(0),\V_{10}(1),\dots$ for $b=10$ as a walk on $\{0,\dots,9\}^{*}$, up to $n=1000$; vertical ellipses compress each block. Within block~$1$ the ten codewords run $1,2,\dots,9,0$, each consecutive pair differing by a single-digit substitution; the steps $\eps\mapsto\texttt{1}$, $\texttt{0}\mapsto\texttt{10}$, and $\texttt{00}\mapsto\texttt{100}$ are single leading-digit insertions taking $0^{k}$ to $1\,0^{k}$. The index $n=1000$ lies in the length-$3$ block. The walk visits every finite decimal string exactly once.}
  \label{fig:walk10}
\end{figure}

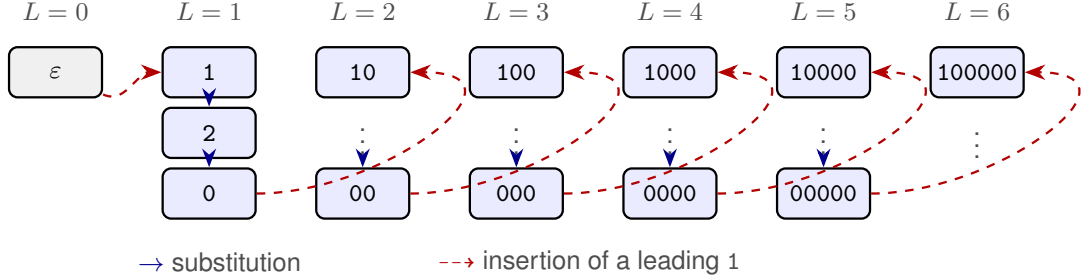
\begin{figure*}[!t]
  \centering
  \resizebox{0.92\textwidth}{!}{%
  \begin{tikzpicture}
    \foreach \x/\l in {0/0,1.9/1,3.8/2,5.7/3,7.6/4,9.5/5,11.4/6}
      \node[blabel] at (\x,0.95) {$\Len=\l$};
    \node[emptynode] (e) at (0,0.2) {$\eps$};
    \node[wordnode] (a1) at (1.9,0.20)  {1};
    \node[wordnode] (a2) at (1.9,-0.55) {2};
    \node[wordnode] (a0) at (1.9,-1.30) {0};
    \node[wordnode] (b1) at (3.8,0.20)  {10};
    \node[blabel]   (bd) at (3.8,-0.55) {$\vdots$};
    \node[wordnode] (b0) at (3.8,-1.30) {00};
    \node[wordnode] (c1) at (5.7,0.20)  {100};
    \node[blabel]   (cd) at (5.7,-0.55) {$\vdots$};
    \node[wordnode] (c0) at (5.7,-1.30) {000};
    \node[wordnode] (d1) at (7.6,0.20)  {1000};
    \node[blabel]   (dd) at (7.6,-0.55) {$\vdots$};
    \node[wordnode] (d0) at (7.6,-1.30) {0000};
    \node[wordnode] (f1) at (9.5,0.20)  {10000};
    \node[blabel]   (fd) at (9.5,-0.55) {$\vdots$};
    \node[wordnode] (f0) at (9.5,-1.30) {00000};
    \node[wordnode] (g1) at (11.4,0.20)  {100000};
    \node[blabel]   (gd) at (11.4,-0.60) {$\vdots$};
    \draw[subEdge] (a1)--(a2); \draw[subEdge] (a2)--(a0);
    \draw[subEdge] (bd)--(b0);
    \draw[subEdge] (cd)--(c0);
    \draw[subEdge] (dd)--(d0);
    \draw[subEdge] (fd)--(f0);
    \draw[insEdge] (e)  to[out=-25,in=180] (a1);
    \draw[insEdge] (a0) to[out=0,in=0,looseness=1.7] (b1);
    \draw[insEdge] (b0) to[out=0,in=0,looseness=1.7] (c1);
    \draw[insEdge] (c0) to[out=0,in=0,looseness=1.7] (d1);
    \draw[insEdge] (d0) to[out=0,in=0,looseness=1.7] (f1);
    \draw[insEdge] (f0) to[out=0,in=0,looseness=1.7] (g1);
    \node[blabel, anchor=west] at (0.9,-2.15) {\textcolor{blue!55!black}{$\rightarrow$} substitution};
    \node[blabel, anchor=west] at (4.6,-2.15) {\textcolor{red!70!black}{$\dashrightarrow$} insertion of a leading \texttt{1}};
  \end{tikzpicture}%
  }%
  \caption{The enumeration $\V_{3}(0),\V_{3}(1),\dots$ for $b=3$ as a walk on $\{0,1,2\}^{*}$, up to $n=1000$; vertical ellipses compress each block. Block~$k$ (the length-$k$ ternary words) runs from $1\,0^{\,k-1}$ to $0^{k}$; consecutive codewords inside a block differ by a single-digit substitution, and each boundary $0^{k}\mapsto1\,0^{k}$ is a single leading-digit insertion. Because $b=3$ is small, the blocks stack quickly: $n=1000$ already lies in the length-$6$ block ($364\le n\le 1092$). The walk visits every finite ternary string exactly once.}
  \label{fig:walk3}
\end{figure*}

\section{Properties}
\label{sec:properties}

Throughout this section $b\ge 2$ is fixed and, for $k\ge 0$,
\[
  B_{k} \;=\; \{\, n\in\NN : \Len(\V_{b}(n))=k \,\}
\]
denotes the \emph{$k$-th length block}. We first determine the length of a codeword, which fixes the blocks explicitly, and then prove the four theorems. Every statement holds for every base $b\ge 2$.

\begin{lemma}[Codeword length]
\label{lem:length}
For every $n\in\NN$, $\Len(\V_{b}(n))=k$, where $k$ is the block index of $n$; equivalently
\[
  \Len(\V_{b}(n))=\big\lfloor\log_{b}\!\big((b-1)n+1\big)\big\rfloor .
\]
Consequently $B_{k}=\{\,n\in\NN : N_{k}\le n\le N_{k+1}-1\,\}$ with $|B_{k}|=b^{k}$, where $N_{k}=(b^{k}-1)/(b-1)$.
\end{lemma}

\begin{proof}
By Definition~\ref{def:V}, $\V_{b}(n)=\shift{\G(r,k)}$ with $r=n-N_{k}$, and the leading increment does not change length, so $\Len(\V_{b}(n))=\Len(\G(r,k))=k$. Since the blocks partition $\NN$, $n\in B_{k}$ iff $N_{k}\le n<N_{k+1}$; multiplying by $b-1>0$ and adding $1$ turns this into $b^{k}\le(b-1)n+1<b^{k+1}$, i.e. $k\le\log_{b}((b-1)n+1)<k+1$, so $k=\lfloor\log_{b}((b-1)n+1)\rfloor$. Finally $|B_{k}|=N_{k+1}-N_{k}=b^{k}$.
\end{proof}

\subsection{Bijectivity}

\begin{lemma}[Block bijection]
\label{lem:block}
For every $k\ge 0$, the restriction of $\V_{b}$ to $B_{k}$ is a bijection onto $\Ab^{k}$.
\end{lemma}

\begin{proof}
On $B_{k}$ the offset $r=n-N_{k}$ ranges over $\{0,\dots,b^{k}-1\}$, and by~(P1) the map $r\mapsto\G(r,k)$ is a bijection onto $\Ab^{k}$. The operator $\shift{\cdot}$ applies the bijection $x\mapsto(x+1)\bmod b$ of $\Ab$ to the leading digit and fixes the remaining digits, so it is a bijection of $\Ab^{k}$ onto itself. Composing, $n\mapsto\V_{b}(n)=\shift{\G(n-N_{k},k)}$ is a bijection from $B_{k}$ onto $\Ab^{k}$. For $k=0$, $B_{0}=\{0\}$ maps to $\Ab^{0}=\{\eps\}$.
\end{proof}

\begin{theorem}[Bijectivity]
\label{thm:bij}
$\V_{b}:\NN\to\Abs$ is a bijection. Equivalently, every finite word over $\Ab$ is the codeword of exactly one natural number, the empty string being the codeword of $0$.
\end{theorem}

\begin{proof}
The blocks $\{B_{k}\}_{k\ge 0}$ partition $\NN$, since by Lemma~\ref{lem:length} each $n$ lies in $B_{k}$ for the unique $k$ with $N_{k}\le n<N_{k+1}$. The sets $\{\Ab^{k}\}_{k\ge 0}$ partition $\Abs$ by length. By Lemma~\ref{lem:block}, $\V_{b}$ maps each $B_{k}$ bijectively onto $\Ab^{k}$. A map that bijects the blocks of one partition onto the corresponding blocks of another is a bijection between the underlying sets; hence $\V_{b}:\NN\to\Abs$ is a bijection, with $\V_{b}(0)=\eps$.
\end{proof}

\subsection{Unit Levenshtein Distance Between Consecutive Codewords}

We isolate the codewords at the two ends of each block.

\begin{lemma}[Block-end codewords]
\label{lem:ends}
For every $k\ge 0$,
\[
  \V_{b}(N_{k+1}-1)=0^{k}
  \qquad\text{and}\qquad
  \V_{b}(N_{k+1})=1\,0^{k},
\]
and, for $k\ge 1$, $\V_{b}(N_{k})=1\,0^{\,k-1}$ \textup{(}with the convention $0^{0}=\eps$\textup{)}.
\end{lemma}

\begin{proof}
For $k\ge 1$ the index $n=N_{k}$ has offset $r=0$, so $\V_{b}(N_{k})=\shift{\G(0,k)}=\shift{0^{k}}=1\,0^{\,k-1}$ by~(P5). For $n=N_{k+1}-1$ the offset is $r=(N_{k+1}-1)-N_{k}=b^{k}-1$, so by~(P5) $\V_{b}(N_{k+1}-1)=\shift{\G(b^{k}-1,k)}=\shift{(b-1)0^{\,k-1}}=0^{k}$; for $k=0$ this reads $\V_{b}(0)=\eps=0^{0}$, which holds by definition. Finally $\V_{b}(N_{k+1})=1\,0^{k}$ is the first identity applied to $k+1$.
\end{proof}

\begin{theorem}[Unit Levenshtein distance]
\label{thm:lev}
For every $n\in\NN$, $\Lev\!\big(\V_{b}(n),\V_{b}(n+1)\big)=1$.
\end{theorem}

\begin{proof}
By Lemma~\ref{lem:length}, $n$ is either interior to its block or at its right end.

\emph{Interior case: $n$ and $n+1$ lie in the same block $B_{k}$}, that is $n\ne N_{k+1}-1$. Their offsets are $r$ and $r+1$ in $\{0,\dots,b^{k}-1\}$, so by~(P2) the words $\G(r,k)$ and $\G(r+1,k)$ differ in exactly one position. Applying $\shift{\cdot}$ adds the same amount ($+1$ modulo $b$) to the leading digit of both, which can neither create nor destroy a disagreement: if the two words agree in the leading position they still do, and if they differ there, values $x\ne x'$ give $(x+1)\ne(x'+1)\bmod b$. Hence $\V_{b}(n)$ and $\V_{b}(n+1)$ have equal length $k$ and differ in exactly one position, so $\Ham(\V_{b}(n),\V_{b}(n+1))=1$ and $\Lev\le 1$; injectivity (Theorem~\ref{thm:bij}) rules out $0$, so the distance is $1$.

\emph{Boundary case: $n=N_{k+1}-1$ and $n+1=N_{k+1}$ for some $k\ge 0$.} By Lemma~\ref{lem:ends}, $\V_{b}(n)=0^{k}$ and $\V_{b}(n+1)=1\,0^{k}$. Inserting a single leading digit $1$ turns $0^{k}$ into $1\,0^{k}$, so $\Lev\le 1$; the two words have different lengths, hence are distinct, giving $\Lev\ge 1$. Thus the distance is $1$.
\end{proof}

\begin{remark}
The proof shows more than the stated bound: consecutive codewords within a block are one \emph{substitution} apart (equal length, one digit changed by $+1$ modulo $b$), whereas each block boundary is a single \emph{insertion} of a leading one, $0^{k}\mapsto1\,0^{k}$. The enumeration $\V_{b}(0),\V_{b}(1),\dots$ is thus a Hamiltonian walk of $\Abs$ under unit edit steps that alternates runs of substitutions---each run a $b$-ary Gray code over $\Ab^{k}$---with a solitary insertion at every index $N_{k}$; Figs.~\ref{fig:walk2}, \ref{fig:walk10} and~\ref{fig:walk3} depict $b=2$, $b=10$, and $b=3$.
\end{remark}

\subsection{Monotone Codeword Length}

\begin{theorem}[Monotone length]
\label{thm:mono}
$\Len(\V_{b}(n))=\lfloor\log_{b}((b-1)n+1)\rfloor$ is monotone non-decreasing: if $n\le n'$ then $\Len(\V_{b}(n))\le\Len(\V_{b}(n'))$.
\end{theorem}

\begin{proof}
By Lemma~\ref{lem:length}, $\Len(\V_{b}(n))=\lfloor\log_{b}((b-1)n+1)\rfloor$. The map $n\mapsto(b-1)n+1$ is increasing and $\log_{b}$ is increasing, so $n\mapsto\log_{b}((b-1)n+1)$ is non-decreasing; the floor of a non-decreasing function is non-decreasing, giving the claim.
\end{proof}

\subsection{The Length-Increment Step}

\begin{theorem}[Length-increment step]
\label{thm:step}
For every $k\ge 1$,
\[
  \V_{b}\!\left(\,\sum_{i=0}^{k-1} b^{i}\right)=\V_{b}\!\left(\frac{b^{k}-1}{b-1}\right)=1\,0^{\,k-1}.
\]
\end{theorem}

\begin{proof}
The finite geometric sum equals $\sum_{i=0}^{k-1}b^{i}=(b^{k}-1)/(b-1)=N_{k}$, so the two arguments coincide, and $\V_{b}(N_{k})=1\,0^{\,k-1}$ by Lemma~\ref{lem:ends}.
\end{proof}

\begin{corollary}[Block structure]
\label{cor:blocks}
For each $k\ge 1$, the codewords of length $k$ are exactly $\{\V_{b}(n): N_{k}\le n\le N_{k+1}-1\}=\Ab^{k}$; the smallest such index, $n=N_{k}=\sum_{i=0}^{k-1}b^{i}$, has codeword $\V_{b}(N_{k})=1\,0^{\,k-1}$, and the largest, $n=N_{k+1}-1$, has codeword $\V_{b}(N_{k+1}-1)=0^{k}$.
\end{corollary}

\begin{proof}
Immediate from Lemmas~\ref{lem:length},~\ref{lem:block} and~\ref{lem:ends} together with Theorem~\ref{thm:step} for the smallest index.
\end{proof}

Theorem~\ref{thm:step} identifies the exact moment the code lengthens: the first codeword of length $k$ appears at $n=N_{k}=(b^{k}-1)/(b-1)$ and equals $1\,0^{\,k-1}$, and Theorem~\ref{thm:lev} guarantees that reaching it from the last length-$(k-1)$ codeword $0^{\,k-1}$ costs a single insertion. For $b=2$ these indices are $2^{k}-1$; for $b=10$ they are the repunits $1,11,111,\dots$

\section{Discussion}
\label{sec:discussion}

\subsection{Interpretation}
Theorems~\ref{thm:bij}--\ref{thm:step} together say that $\V_{b}$ is a Gray code for the countably infinite set $\Abs$, ordered by the natural numbers, in which the local move is a unit edit rather than a unit digit change. Classical reflected Gray codes are Hamiltonian paths of a finite $b$-ary cube under Hamming adjacency \cite{gilbert1958gray}; $\V_{b}$ is a Hamiltonian walk of the infinite string space under Levenshtein adjacency, decomposing into cube-like substitution runs---each run a $b$-ary Gray code over $\Ab^{k}$---that are stitched together by single insertions at the indices $N_{k}=(b^{k}-1)/(b-1)$. The construction achieves this without sacrificing bijectivity: unlike a fixed-length code, which caps the range at $b^{w}-1$, the variable-length code uses every finite word over $\Ab$, including the empty one, exactly once. The binary case $b=2$ is the Hamiltonian walk of the hypercube; larger bases trade a taller substitution run within each block ($b^{k}$ words instead of $2^{k}$) for shorter codewords.

\subsection{Compression Relative to Fixed-Length Gray Codes}
Because $\Len(\V_{b}(n))=\lfloor\log_{b}((b-1)n+1)\rfloor$ (Lemma~\ref{lem:length}), the code spends only as many digits as the magnitude of the integer warrants, whereas a fixed $w$-digit $b$-ary Gray code spends $w$ digits on every integer it represents and represents only $\{0,\dots,b^{w}-1\}$. To store all integers in $\{0,1,\dots,N\}$, the fixed code must choose $w=\lceil\log_{b}(N+1)\rceil$, for a total of $(N+1)\,w$ digits. Over a range that fills blocks $0,\dots,K$ exactly, that is $N=N_{K+1}-1$, the variable-length code instead uses
\begin{equation*}
  \sum_{n=0}^{N_{K+1}-1}\Len(\V_{b}(n))=\sum_{k=0}^{K}k\,b^{k}
  =\frac{K\,b^{K+2}-(K+1)\,b^{K+1}+b}{(b-1)^{2}},
\end{equation*}
which is smaller and, per integer, is within one digit of the ideal $\log_{b}$ of the magnitude. Figures~\ref{fig:comp2}, \ref{fig:comp10} and~\ref{fig:comp3} contrast the two length profiles and the cumulative storage for $b=2$, $b=10$, and $b=3$, each taken over the range $\{0,\dots,b^{w}-1\}$ that a fixed $w$-digit code represents ($w=5,3,6$, i.e.\ $n$ up to $31$, $999$, $728$). The saving is largest exactly where it is most often needed, on the small integers that dominate index-like and counter-like data; it is bought at the usual price of a variable-length code, namely that codewords are not, on their own, self-delimiting. Being a bijective rather than a prefix-free code, $\V_{b}$ places every finite word in use, so a boundary between adjacent codewords in a concatenated stream must be supplied by context or by composition with a self-delimiting envelope such as an Elias code \cite{elias1975universal} or a Fibonacci-based code \cite{apostolico1987robust,fraenkel1996robust}; in exchange it attains the shortest possible length profile that a monotone complete code can have, since the $b^{k}$ words of length $k$ are all used before any word of length $k+1$.

\begin{figure*}[!t]
  \centering
  \includegraphics[width=\textwidth]{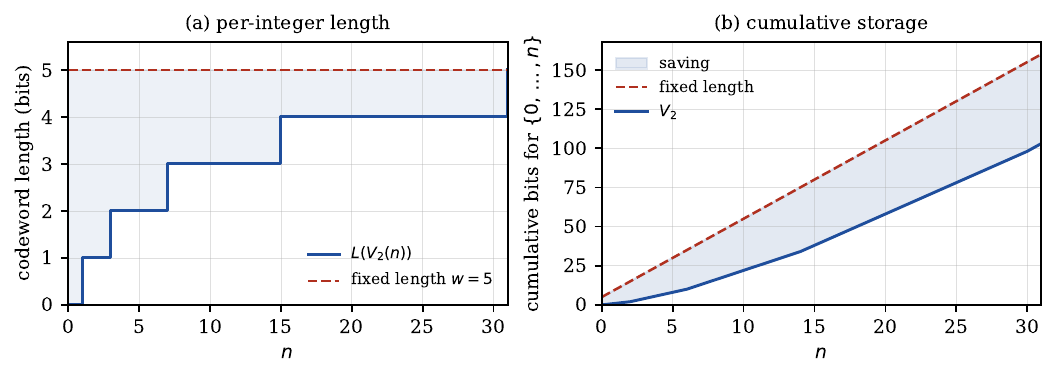}
  \caption{Compression of the binary code $\V_{2}$ against a fixed-length $5$-bit binary Gray code, over the range $\{0,\dots,2^{5}-1\}=\{0,\dots,31\}$ that such a code represents. (a)~Per-integer codeword length $\Len(\V_{2}(n))=\lfloor\log_{2}(n+1)\rfloor$ (steps) versus the constant $5$ bits (dashed). (b)~Cumulative bits to store $\{0,\dots,n\}$; the shaded gap is the saving, greatest on the small integers.}
  \label{fig:comp2}
\end{figure*}

\begin{figure*}[!t]
  \centering
  \includegraphics[width=\textwidth]{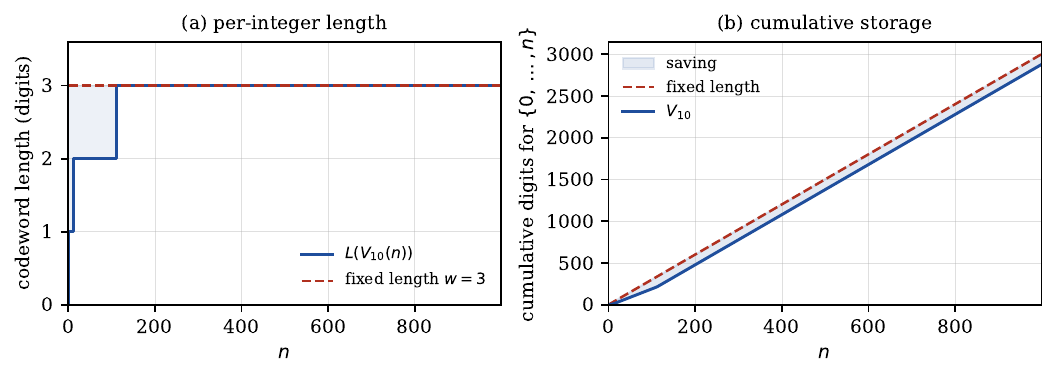}
  \caption{Compression of the decimal code $\V_{10}$ against a fixed-length $3$-digit decimal Gray code, over the range $\{0,\dots,10^{3}-1\}=\{0,\dots,999\}$ that such a code represents. (a)~Per-integer codeword length $\Len(\V_{10}(n))=\lfloor\log_{10}(9n+1)\rfloor$ (steps) versus the constant $3$ digits (dashed). (b)~Cumulative digits to store $\{0,\dots,n\}$; the shaded gap is the saving.}
  \label{fig:comp10}
\end{figure*}

\begin{figure*}[!t]
  \centering
  \includegraphics[width=\textwidth]{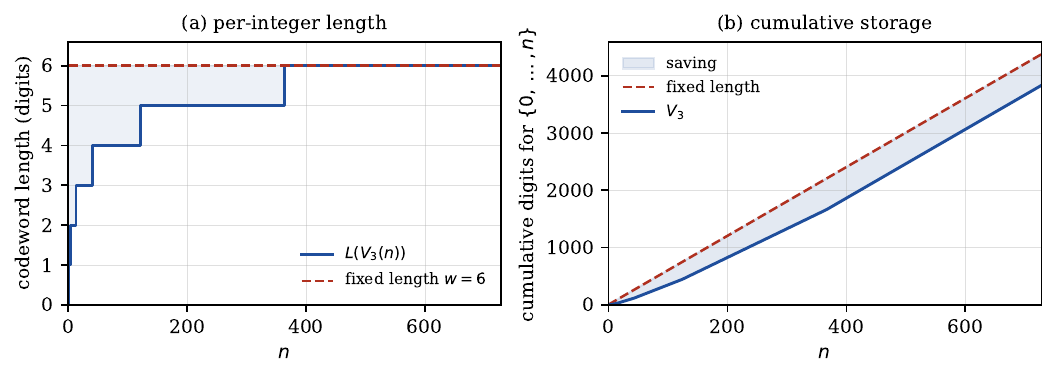}
  \caption{Compression of the ternary code $\V_{3}$ against a fixed-length $6$-digit ternary Gray code, over the range $\{0,\dots,3^{6}-1\}=\{0,\dots,728\}$ that such a code represents. (a)~Per-integer codeword length $\Len(\V_{3}(n))=\lfloor\log_{3}(2n+1)\rfloor$ (steps) versus the constant $6$ digits (dashed). (b)~Cumulative digits to store $\{0,\dots,n\}$; the shaded gap is the saving. The staircase has more, finer steps than at larger bases because ternary codewords grow one digit at a time over shorter ranges.}
  \label{fig:comp3}
\end{figure*}

\subsection{A Complete, Unit-Edit Representation for Numeric Tokens}
\label{subsec:llm}
The properties proved above have a natural application to the way current language models represent the natural numbers they generate. We state it as a design rationale and a testable hypothesis, not as a proven performance claim.

Language models emit numbers as strings of symbols drawn one step at a time from a fixed vocabulary, and the representation used for those symbols has a measurable effect on what is learned \cite{nogueira2021arithmetic,wallace2019numbers,singh2024tokenization,thawani2021representing}. Two structural defects of ordinary base-$b$ notation are relevant, and $\V_{b}$ repairs each by construction.

\emph{Incompleteness and wasted probability mass.} Under the usual convention a natural number has a unique shortest base-$b$ representation only if leading zeros are forbidden; the admissible strings---those with a nonzero leading digit, together with the single string ``$0$''---form a sparse, syntactically constrained subset of $\Ab^{*}$. A generative model over symbol sequences nonetheless assigns probability to \emph{every} string, so at each numeric position it must learn to concentrate its output on the legal forms and to suppress the many illegal ones (spurious leading zeros, and the empty string). That is capacity and training signal spent on syntax rather than on the quantities themselves, and it makes the string-to-value map partial or many-to-one rather than a clean bijection. By Theorem~\ref{thm:bij}, $\V_{b}$ is a bijection onto \emph{all} of $\Abs$: every string the model can produce is the valid, unique codeword of some natural number, so no output is wasted on an invalid form and there is no redundant spelling to normalize away.

\emph{Non-locality of the successor.} In base-$b$ notation consecutive integers can differ in arbitrarily many symbols: the step $b^{k}-1\mapsto b^{k}$ rewrites $k$ digits and lengthens the string. When numerically adjacent values are far apart as strings, the model must learn counting and the successor relation as global rewrites. Under $\V_{b}$, Theorem~\ref{thm:lev} guarantees that consecutive integers are exactly one edit apart---one symbol substituted within a length block, one symbol inserted at a block boundary---so counting becomes a local operation on the codeword and the successor map has minimal descriptive complexity. This is the same mechanism that makes classical Gray codes robust in position encoders and analog-to-digital converters \cite{gray1953pulse,doran2007gray}: a single-symbol slip decodes to a neighboring integer rather than a distant one, so representation errors degrade gracefully in magnitude instead of catastrophically---a desirable property when a sampled numeral may contain an occasional wrong symbol.

Two further features fit the setting. The base $b$ is a free parameter (Section~\ref{sec:method}), so the alphabet can be matched to a chosen numeric vocabulary: digit-level ($b=10$) encoding, which several production models already adopt \cite{singh2024tokenization}, or a larger $b$ to shorten codewords. And the length profile is near-optimal, $\Len(\V_{b}(n))=\lfloor\log_{b}((b-1)n+1)\rfloor$ (Lemma~\ref{lem:length}), so the compactness gains quantified above apply, with the small integers that dominate index- and counter-like data receiving the shortest codewords.

Two caveats bound the claim. First, $\V_{b}$ is not self-delimiting: in a token stream the extent of a numeral must still be marked, whether by a reserved delimiter, a fixed field width, or composition with a self-delimiting envelope \cite{elias1975universal,fraenkel1996robust}. This is orthogonal to the two properties above and is shared by ordinary digit strings, which are likewise not self-delimiting once embedded in text. Second, whether these structural advantages translate into measurably better learning or generalization is an empirical question: it calls for training models on $\V_{b}$-encoded numerals and comparing against standard sub-word and digit-level baselines on counting, arithmetic, and retrieval of numeric keys. The construction makes such an experiment well defined by supplying, for every base, a complete code whose successor is a unit edit; we leave its empirical evaluation to future work.

\section{Conclusions}
\label{sec:conclusion}

We have introduced, for every integer base $b\ge 2$, a variable-length Gray code $\V_{b}$ that maps the natural numbers onto the set of all finite words over the alphabet $\{0,1,\dots,b-1\}$. The construction orders the numbers into length blocks---block $k$ being the $b^{k}$ words of length $k$, starting at $N_{k}=(b^{k}-1)/(b-1)$---and lists each block along the $b$-ary Gray code of the within-block offset, with its leading digit incremented modulo $b$. We proved four properties, all for arbitrary $b$. The code is a bijection onto $\Abs$, hence a complete code that uses every word exactly once, including the empty word (Theorem~\ref{thm:bij}). Consecutive integers have codewords at unit Levenshtein distance---a single-digit substitution inside each length block and a single leading-digit insertion $0^{k}\mapsto1\,0^{k}$ at every block boundary---so the enumeration is a Hamiltonian walk of the string space under unit edit steps (Theorem~\ref{thm:lev}). Codeword length is monotone non-decreasing and equals $\lfloor\log_{b}((b-1)n+1)\rfloor$ (Theorem~\ref{thm:mono}). The code lengthens exactly at $n=N_{k}=\sum_{i=0}^{k-1}b^{i}$, where the codeword is $1\,0^{\,k-1}$ (Theorem~\ref{thm:step}). The binary code ($b=2$) recovers the reflected Gray-code bijective numeration, and the decimal code ($b=10$) served as a second example; using the modular $b$-ary Gray code, whose block ends at $(b-1)0^{\,k-1}$ for every $b$, is what lets a single leading-digit insertion join consecutive blocks at every base.

The result is a family of codes that keep the single-change spirit of the classical Gray code while lifting both its fixed-length and its binary-alphabet restrictions, gaining an unbounded range, a near-optimal self-adapting length, and completeness over $\Abs$ at every base.

Beyond its combinatorial interest, the code is motivated by a practical question in machine learning: language models generate natural numbers constantly, yet the positional notations they use are neither complete---leading-zero strings are invalid or redundant---nor locally stable under increment. A code that is a bijection onto all finite strings and moves by a single edit between consecutive integers removes both defects, making every generated string a valid numeral and turning the successor relation into a local operation. We have argued (Section~\ref{subsec:llm}) that these are desirable properties for representing numeric tokens, and we identify the empirical evaluation of $\V_{b}$-encoded numerals in language-model training---against standard sub-word and digit-level tokenizations, on counting, arithmetic, and numeric-key retrieval---as the natural next step.


\bibliographystyle{IEEEtran}
\bibliography{references}

\vfill

\end{document}